\documentclass{article}

\usepackage{arxiv}

\usepackage[utf8]{inputenc} 
\usepackage[T1]{fontenc}    
\usepackage{hyperref}       
\usepackage{url}            
\usepackage{booktabs}       
\usepackage{amsfonts}       
\usepackage{nicefrac}       
\usepackage{microtype}      
\usepackage{lipsum}		
\usepackage{graphicx}

\usepackage{amsmath}
\usepackage{amsthm}
\usepackage{amsfonts}
\usepackage[numbered]{bookmark}
\usepackage{float}
\usepackage{graphicx}
\usepackage{caption}
\usepackage{subcaption}
\usepackage{lmodern}
\usepackage{textcomp}
\usepackage{lscape}
\usepackage{enumitem}
\usepackage{booktabs}
\usepackage{textpos}
\usepackage{chngcntr}

\hypersetup{
  colorlinks  = true, 
  urlcolor    = blue, 
  linkcolor   = blue, 
  citecolor   = red   
}

\newcommand{\E}{\mathbb{E}}
\newcommand{\Q}{\mathbb{Q}}
\newcommand{\cov}{\mathbb{C}\mathrm{ov}}

\newcommand{\mc}{\mathrm{mc}}
\DeclareMathOperator{\tr}{tr}
\DeclareMathOperator{\rank}{rank}

\newtheorem{thm}{Theorem}[section]
\newtheorem*{thm*}{Theorem}
\theoremstyle{remark}

\newtheorem*{rem*}{Remark}

\renewcommand{\theequation}{\thesection.\arabic{equation}}
\counterwithin*{equation}{section}

\title{Efficient Least Squares Monte-Carlo Technique for PFE/EE Calculations}

\author{ 
        \href{https://orcid.org/0000-0001-6204-5819}{\includegraphics[scale=0.06]{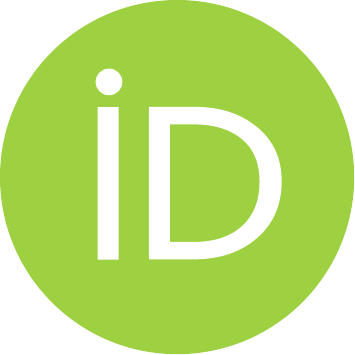}\hspace{1mm}Yuriy Krepkiy}\thanks{The results contained in this paper were obtained when the author was working for IBM Canada.}\\
	RBC Capital Markets\\
	Toronto, Canada \\
	\texttt{Yuriy.Krepkiy@rbccm.com} \\
	\And
        \href{https://orcid.org/0000-0002-3509-230X}{\includegraphics[scale=0.06]{orcid.pdf}\hspace{1mm}Asif Lakhany} \\
	SS\&C Algorithmics\\
	Toronto, Canada \\
	\texttt{Asif.Lakhany@sscinc.com} \\
	\And
        \href{https://orcid.org/0000-0003-0247-8107}{\includegraphics[scale=0.06]{orcid.pdf}\hspace{1mm}Amber Zhang}\\
	SS\&C Algorithmics\\
	Toronto, Canada \\
	\texttt{Amber.Zhang@sscinc.com} \\
	}

\renewcommand{\headeright}{} 
\renewcommand{\undertitle}{} 

\hypersetup{
pdftitle={Efficient Least Squares Monte-Carlo Technique for PFE/EE Calculations},
pdfsubject={q-fin.CP},
pdfauthor={Yuriy Krepkiy, Asif Lakhany, Amber Zhang},
pdfkeywords={LSMC, AMC, Nested MC},
}

\begin{document}
\maketitle

\begin{abstract}
We describe a regression-based method, generally referred to as the Least Squares Monte Carlo (LSMC) method, to speed up exposure calculations of a portfolio. We assume that the portfolio contains several exotic derivatives that are priced using Monte-Carlo on each real world scenario and time step. Such a setting is often referred to as a Monte Carlo over a Monte Carlo or a Nested Monte Carlo method.
\end{abstract}

\keywords{American Monte-Carlo \and Least-Squares Monte-Carlo \and AMC \and LSMC \and Nested Simulation}

\section{Introduction} 
\label{sec:Introduction} 
Least Squares Monte-Carlo (LSMC) is a technique based on least squares
regression, which we describe in this paper. We think of LSMC as a special case
of a larger class of methods that are referred to as American Monte-Carlo in the
literature \cite{ghamami_zhang_eff_amc}, \cite{sokol_xva}. The term AMC has
its origins in the work of Longstaff and Schwartz \cite{amc_shwartz}. Therein, the
authors describe a method to price American Options that relies on building a
conditional expectation function using a least squares regression technique
over a set of explanatory variables.  In the simplest case, the set of
explanatory variables would include the current state of the underlying risk
factors. In its original form, the method uses two sets of Monte-Carlo
simulations. One simulation is used to build the conditional expectation
function by regressing over the stock price and the indicator that the option is
in the money using backward propagation of state variables. Once the backward
pass is done, a different set of Monte-Carlo paths are used to move forward to
price the instrument. During the forward pass, the already constructed
conditional expectation functions are used at every observation date to
determine whether it is optimal to exercise or not.

In its modern usage, in the context of risk management, the term AMC (or LSMC)
is used to handle the Nested Monte-Carlo (NMC). As an illustration, consider
Potential Future Exposure (PFE) or Expected Exposure (EE) for a
portfolio consisting of exotic instruments.\footnote{For detailed example of
    PFE and EE see \cite{xva_ruiz15}.} 
To estimate an EE and PFE profiles, we value the portfolio on a set of future
market, or ``outer'', scenarios generated across time. Suppose we use 5,000 such
scenarios, and use 5,000 risk-neutral, ``inner'', paths to obtain a single
price estimate on each outer scenario. Under this set-up, the computational cost for each exotic
instrument will be proportional to $5,000 \times 5,000 = 25,000,000$ on each
time step. Clearly, we need to apply some clever techniques to reduce the
computation cost. 

One such solution has been proposed by Barrie and Hibbert
\cite{barrie_hibbert11}.  Berrie and Hibbert's method reduces the overall
computational cost by decreasing the number of inner paths.  Their method aims
at reducing Monte-Carlo errors by regressing the estimated prices against a set
of explanatory variables generated in the outer loop. We refer to price
estimates obtained using this method as LSMC price estimates. To illustrate the
LSMC method, consider a vanilla call option under Geometric Brownian Motion
with constant drift and volatility. We use a GBM model in the outer loop to
generate a set of underlying stock prices and use GBM model within the inner
loop to price the option. Assume that the option matures in one year (time step
360), and we are interested in computing prices at time steps 15, 30, \dots,
360. On each time step, in a typical NMC setting, we use 5,000 inner paths to
obtain MC price estimates. Under the LSMC setting, we use 30 inner paths, for
example, per outer scenario to obtain a set of MC prices on a given time step.
When these prices are plotted against 5,000 underlying (outer) spot values on
a given time step, we expect to see some relationship. We assume that this
relationship can be explained by 
\begin{equation}
    \label{eq:gbm_example}
    y_{\mc,\, i} = \beta_0 + \beta_1 S_i + \beta_2 S_i^2 + \beta_3 S_i^3 + \xi_i
\end{equation}
where $y_{\mathrm{\mc},\, i}$ is the MC price, $S_i$ is the spot price
of the equity, and $\xi_i$ is the error on the $i^{\mathrm{th}}$ scenario.

Once the preliminary $5,000 \times 30$ simulation is done, we have 5,000 MC
price estimates and 5,000 spot values that we use to build regression model
\eqref{eq:gbm_example} to obtain coefficient estimates $\{\hat{\beta}_{i}\}_{i =
0}^3$. These coefficient estimates are then used to obtain LSMC price estimates
given by
\begin{equation*}
    \hat{y}_{\mathrm{\mc},\, i} = \hat{\beta}_0 + \hat{\beta}_1 S_i + \hat{\beta}_2 S_i^2 +
    \hat{\beta}_3 S_i^3.
\end{equation*}

Figure \ref{fig:gbm_example} compares the LSMC price estimates to the
Black-Scholes analytic prices. The graphs are provided for time steps 15 and
345. In order to obtain the ``hockey stick'' shape near maturity of the option,
we incorporate a dummy variable
    \[d_i = \mathbb{I}\{S_i \geq K\},\quad \text{where $K$ is the strike price}\]
into equation \eqref{eq:gbm_example} to yield
\begin{equation}
    \label{eq:gbm_example_dummy}
    y_{\mathrm{\mc},\, i} = \beta_0 + \beta_1 S_i + \beta_2 S_i^2 + \beta_3 S_i^3 + d_i +
    \beta_4 d_i S_i + \beta_5 d_i S_i^2 + \beta_6 d_i S_i^3 + \xi_i.
\end{equation}

\begin{figure}[hp]
    \centering
    \caption{GBM Call Example: Model \ref{eq:gbm_example} is used on the left,
    and model \ref{eq:gbm_example_dummy} is used on the right.}
    \includegraphics[width=\textwidth]{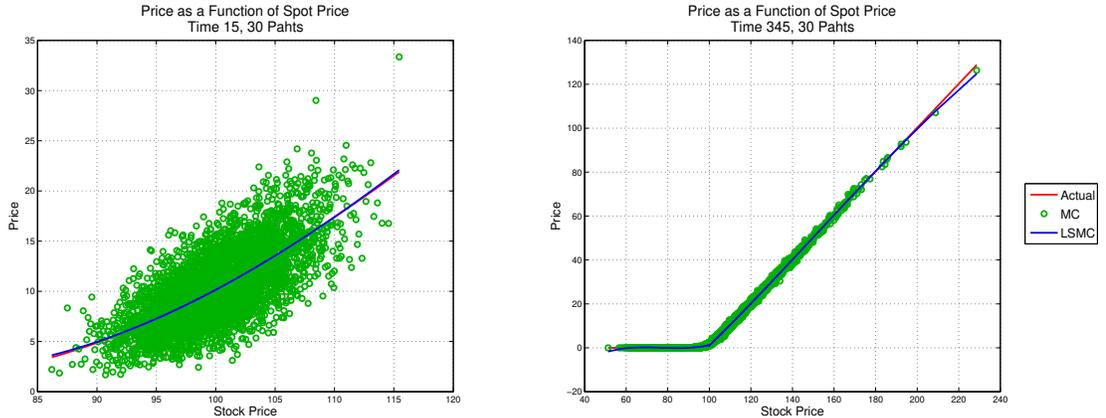}
    \label{fig:gbm_example}
\end{figure}

Figure \ref{fig:gbm_example} reveals that the LSMC price estimates are much
closer to the true analytic prices than the original MC prices.  Interestingly
enough, Barrie and Hibbert's method can be viewed as a variance reduction
technique, though it was not presented as such in their original report. In this
paper we look at the method more closely and present a variance formulation that
could be used to evaluate the accuracy of price estimates. We also look at the
PFE and EE profiles for a set of exotic instruments to test the accuracy of this
method. 

\section{Least Squares Monte-Carlo} 
\label{sec:Least_Squares_Monte-Carlo}
We assume, at outset, that $n$ Monte-Carlo scenarios have been generated over
$k$ time points in the set $\mathcal{T} := \{t_1,\, t_2,\, \dots,\, t_k\}$,
where $t_i < t_j$, for $i < j$. These time points may represent important dates
such as payment dates, fixing dates, etc. At any time $t \in
\mathcal{T}$, we have $n$ vectors $\mathbf{x}_{t,\, i}$, for $i = 1,\, 2\,
\dots,\, n$, representing the cross sectional information obtained from the
scenarios at time $t$ and absorbing all the information up to time $t$. These
vectors form the explanatory variables for the regression model for time $t
\in \mathcal{T}$. 

Suppose that we are interested in obtaining the time $t$ price of an exotic
instrument maturing at time $T$. If $f(\mathbf{x}_{t,\, i})$ denotes the
discounted payoff of this instrument, then the time $t$ price is obtained under
the conditional expectation

\begin{equation} 
    \label{eq:price} 
    y_i := \E^\Q \Big[ f(\mathbf{x}_{t,\, i}) \, \Big|\, \mathcal{F}_t \Big] .
\end{equation}

One can estimate the integral on the right-hand-size of equation
\eqref{eq:price} using a MC method. Doing so, will generate a sequence of MC
estimates $y_{\mc,\, i}$, $i = 1,\, 2,\, \dots,\, n$ at each time
$t \in \mathcal{T}$. These prices form the corresponding set of response
variables for our regression model. We can write the relationship as
\begin{equation}
    \label{eq:mc_approx}
    y_{\mc,\, i} = \sum_{j=0}^m \beta_j\, b_j(\mathbf{x}_{t,\, i}) + \xi_{i} 
\end{equation}
where $\{\beta_{j}\}_j^m$ are coefficients of expansion (as in equation
\eqref{eq:gbm_example})) and $\{b_j(\mathbf{x}_{t,\, i}\}_1^m$ are basis functions (as
$1,\, S_i,\, S_i^{2} \dots$ in equation \eqref{eq:gbm_example}). We obtain MC
prices by computing the mean of $p$ generated discounted payoffs in the inner loop.
Specifically, 
\begin{equation}
    \label{eq:mc_price}
    y_{\mc,\, i} = \frac{1}{p} \sum_{k = 1}^{p}{ f_{k}(\mathbf{x}_{t,\, i})}, 
\end{equation}
where $\mathbf{x}_{t,\, i}$ is generated under risk neutral measure, and
$f_k(\cdot)$ is the value of the discounted payoff function on the $k^{th}$ path.
The Monte-Carlo price can represented as the sum of the true value $y_i$ and a
Monte-Carlo error $\xi_{\mc,\, i}$, or
\begin{equation} \label{eq:mcerr}
    y_{\mc,\, i} = y_i + \xi_{\mc,\, i} \quad.
\end{equation}

Now, suppose we let $p \to \infty$ in equation \eqref{eq:mc_approx}, then 
\begin{align}
    \lim_{p \to \infty} y_{\mc, \,i}  
    & = \lim_{p \to \infty}\sum_{j=0}^m \beta_j\, b_j(\mathbf{x}_{t,\, i}) + \xi_{i} \nonumber \\
    & = \sum_{j=0}^m \beta_j\, b_j(\mathbf{x}_{t,\, i}) + \lim_{p \to \infty} \xi_{i} \quad.  \label{eq:lim} 
\end{align}
As the number of paths tends to infinity, $\xi_{i}$ will approach
$\xi_{\mathrm{d},\, i}$.  \footnote{$\lim_{p \to \infty}{y_{\mc,\, i}} =
y_i$ by the Strong Law of Large Numbers. } In case of a perfect model,
$\xi_{\mathrm{d},\, i}$ equals zero. Generally, $\xi_{\mathrm{d},\, i}$ is
$\mathcal{F}_t$ measurable (deterministic) and is often a function of
$\mathbf{x}_{t,\, i}$. From equation \eqref{eq:lim}, we conclude
that
\begin{equation}
    y_{i} = \sum_{j=0}^m \beta_j\, b_j(\mathbf{x}_{t,\, i}) + \xi_{\mathrm{d},\, i}.
    \label{eq:mc_approxl} 
\end{equation}
Substituting the right-hand-side of equation \eqref{eq:mc_approxl} into equation
\eqref{eq:mcerr} yields the final model
\begin{equation*}
    \label{eq:final_model}
    y_{\mc,\, i} = \sum_{j=0}^m \beta_j\, b_j(\mathbf{X}_{t,\, i}) +
    \xi_{\mathrm{d},\, i}  + \xi_{\mc,\, i} \quad.
\end{equation*}
The total error $\xi_{i}$ from equation \eqref{eq:mc_approx} can be viewed as
the sum of deterministic part $\xi_{\mathrm{d},\, i}$ and a random part
$\xi_{\mc,\, i}$. In practice, we do not observe $\xi_{\mathrm{d},\, i}$
directly nor can accurately measure $\xi_{\mc,\, i}$ unless we use high
number of paths. However, we can estimate the variance of $\xi_{\mc\, i}$ using
the discounted payoff. 

For sufficiently large number of paths, $\mathbf{y}_\mc$ follows
normal distribution with mean $\bf{0}$ and $\bf{\Sigma}$. The non-biased
estimator for $\boldsymbol{\Sigma}_{\mc}$ then is given by
\begin{equation}
    \label{eq:var_est}
    [\hat{\boldsymbol \Sigma}_{\mc}]_{ij} = \frac{1}{p(p - 1)} \sum_{k =
    1}^{p}(f_k(\mathbf{X}_{T,\, i}) - y_{\mc,\,
i})(f_k(\mathbf{X}_{T,\, j}) - y_{\mc,\, j}).
\end{equation}
In case $\boldsymbol \Sigma_{\mc} \equiv \sigma_{\mc}^{2}
\mathbf{I}$, the MC variance estimator becomes
\begin{equation*}
    \hat{\sigma}_{\mc}^2 = \frac{1}{n} \sum_{j =
    1}^{n}[\hat{\boldsymbol \Sigma}_{\mc}]_{jj}\,.
\end{equation*}

In a matrix notation, we can write the regression model \eqref{eq:mc_approx} as
\begin{align} 
    \label{eq:reg_mc} 
    \mathbf{y}_{\mc} &= \mathbf{X} \boldsymbol{\beta} +
    \boldsymbol{\xi}.
\end{align} 
then, the coefficient estimator becomes 
$$\hat{\boldsymbol{\beta}} =
(\mathbf{X^TX})^{-1} \mathbf{X^Ty_{\mc}},$$ 
and 
$$\hat{\bf{y}}_{\mc} = \bf{H}
\bf{y}_{\mc}$$  
with 
$$\bf{H} = \bf{X}(\mathbf{X^TX})^{-1} \mathbf{X^T}.$$
$\bf{H}$ is also known as an orthogonal projection. Conquently, the total 
variance of $\hat{\bf{y}}_\mc$ is lower than the total variance of $\bf{y}_\mc$.
This is given by the following theorem and is proven in appendix
\ref{sec:Linear_Regression}.

\begin{thm}
    \label{thm:var_red}
    Let $\mathbf{Y}: \Omega \rightarrow \mathbb{R}^m$ be a random vector having
    a finite variance, and let $\mathbf{H}$ be an orthogonal projection onto a
    linear subspace of $\mathbb{R}^m$. Then,
    \begin{equation*}
        \tr(\cov(\mathbf{HY})) \leq \tr(\cov(\mathbf{Y})),
    \end{equation*}
    where $\tr(\cdot)$ denotes the trace operator. 
\end{thm}

By Theorem \ref{thm:var_red}, we know that the total variance reduction is
expected under an orthogonal projection
\begin{equation}
    \label{eq:lsmc_var}
    \cov(\hat{\mathbf{Y}}_{\mc}) = \cov(\mathbf{HY}_{\mc}) =
    \mathbf{H}\cov(\mathbf{Y}_{\mc})\mathbf{H}
\end{equation}
where $\cov(\mathbf{Y}_{\mc}) = \boldsymbol{\Sigma}_{\mc}$
represents the covariance matrix of $\boldsymbol{\xi}_{\mc}$ that is estimated
using the discounted payoffs generated in the inner loop. Furthermore,
when the covariance matrix is equal to $\Sigma = \sigma  \bf{I}$  for some
constant $\sigma$ the ratio of total LSMC variance and total MC variance is
given by
\begin{align*}
    \frac{\tr(\cov(\hat{\bf{y}}_{\mc}))} {\tr(\cov(\bf{y}_{\mc}))} &=
    \frac{\rank(\bf{X})}{n} \\
    &= \frac{m}{n}
\end{align*}
which follows from the properties of $\bf{H}$. 

\section{Illustrative Results} 
\label{sec:result}
In this section, we test the accuracy of the LSMC methodology.  An Arithmetic
Asian Option is used to check the fit of various degree polynomials.
After we obtain a good fit, we look at the EE and PFE profiles of a Barrier
Option, Target Accrual Redemption Note, Accumulating Forward Contract and Asian
Option. For the reader's convenience, we describe the payoffs of these
instruments in appendix \ref{app:def_inst}. 

\subsection{LSMC Analysis}
\label{subsec:lsmc_analysis}
Consider an Arithmetic Asian Put Option that matures at time $T$ with $s$ time
fixings $\{t_i\}_{i = 1}^{s}$ and weights $\{w_i\}_{i = 1}^{s}$.  The payoff at
maturity takes the form 
\begin{equation}
    \text{payoff} = (K  - \sum_{i = 1}^{s} { w_i S_{t_i} })^{+},\,
    \text{ $0 \leq t_i \leq T\quad \forall i$ }, 
    \label{eq:asian_option}
\end{equation}
where $S_{t_{i}}$ is the price of the underlying equity at time $t_i$. For
simplicity, assume that $S$ follows a GBM with constant drift and volatility in
the outer and inner loops.\footnote{The drift in the outer loop is set at 0.1,
while the drift in the inner loop is set at 0.05.} Let the observation time be
$t_k$ for $1 < k \leq s$.  We consider orthogonal Forsythe polynomials
\cite{forsythe_poly} for basis functions using $B_{t_k} = \sum_{i}^k w_i  S_i$
as the explanatory variable in the regression model. We perform the following
steps to obtain LSMC estimates:
\begin{itemize}
    \item Transform the explanatory variables to take values between $-1$ and
        $1$ using
        \begin{equation*}
            B_{t_k,\, j}^{*} := \frac{2 B_{t_k,\, j} - \max(\mathbf{B}_{t_k}) -
            \min(\mathbf{B}_{t_k})}{\max(\mathbf{B}_{t_k}) - \min(\mathbf{B}_{t_k})} 
        \end{equation*}
    \item Obtain regression matrix $\mathbf{X}$ using Forsythe polynomial
        expansion \nocite{forsythe_poly}
    \item Estimate the coefficients $\boldsymbol{\beta}$ in model \eqref{eq:reg_mc}
    \item Use the coefficient estimates to obtain LSMC prices
\end{itemize}
We set a numerical experiment with 1, 10, 30, 50, 100 and 10,000 inner paths
to compute and regress $\mathbf{Y}_{\mc}$ against $\mathbf{X}$ to get
$\hat{\mathbf{Y}}_{\mathrm{\mc}}$. To study the accuracy of the fit and
variance reduction, we compare LSMC estimates to MC prices obtained using
$131,072$ Sobol paths. The results are summarized in Figures\ref{fig:fit}
\ref{fig:mc_var} and \ref{fig:lsmc_red}.

\begin{figure}[hp]
    \centering
    \captionsetup{justification=centering}
    \caption{MC vs. LSMC Price: \\ First time step fit using polynomial of
    degree five}
    \includegraphics[width=\textwidth]{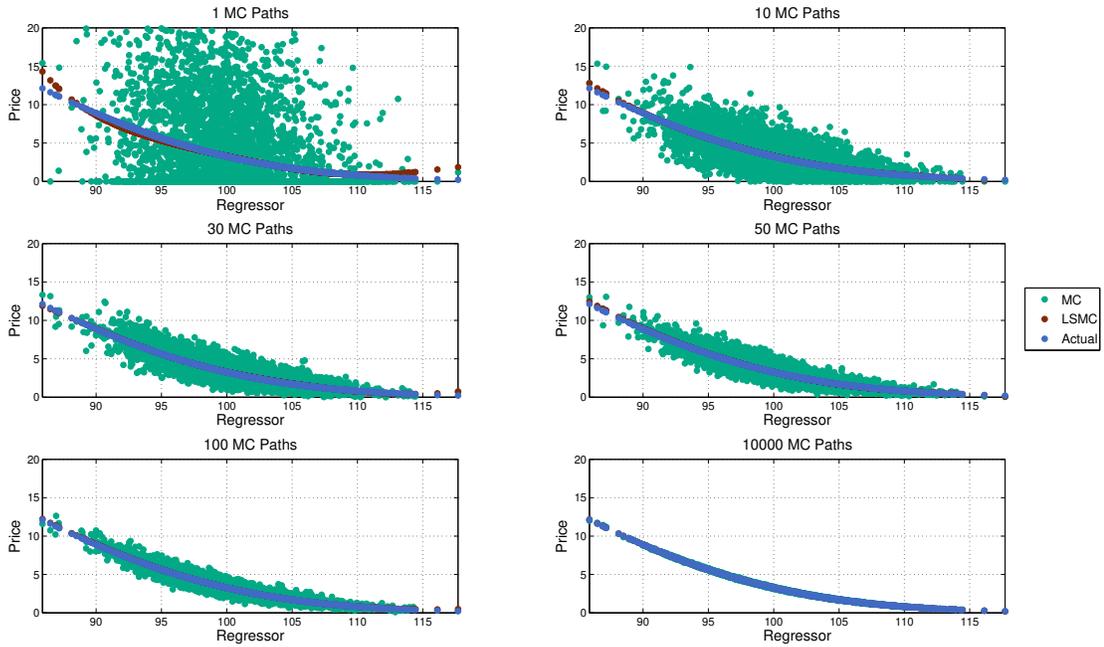}
    \label{fig:fit}
\end{figure}
Note that fully estimated covariance matrix was used to obtain the ratio of
total of MC variance and total of LSMC variance. As one can see, the ratio of
total of MC to LSMC variance is close to $\mathrm{rank}(\mathbf{X}) / n  = 6 /
5000 = 0.0012$. That is, using regression, we are able to reduce the
variance by about 99.9\%.

\begin{figure}
    \centering
    \captionsetup{justification=centering}
    \caption{Monte-Carlo Variance Density: \\First time step}
    \includegraphics[width=0.75\textwidth]{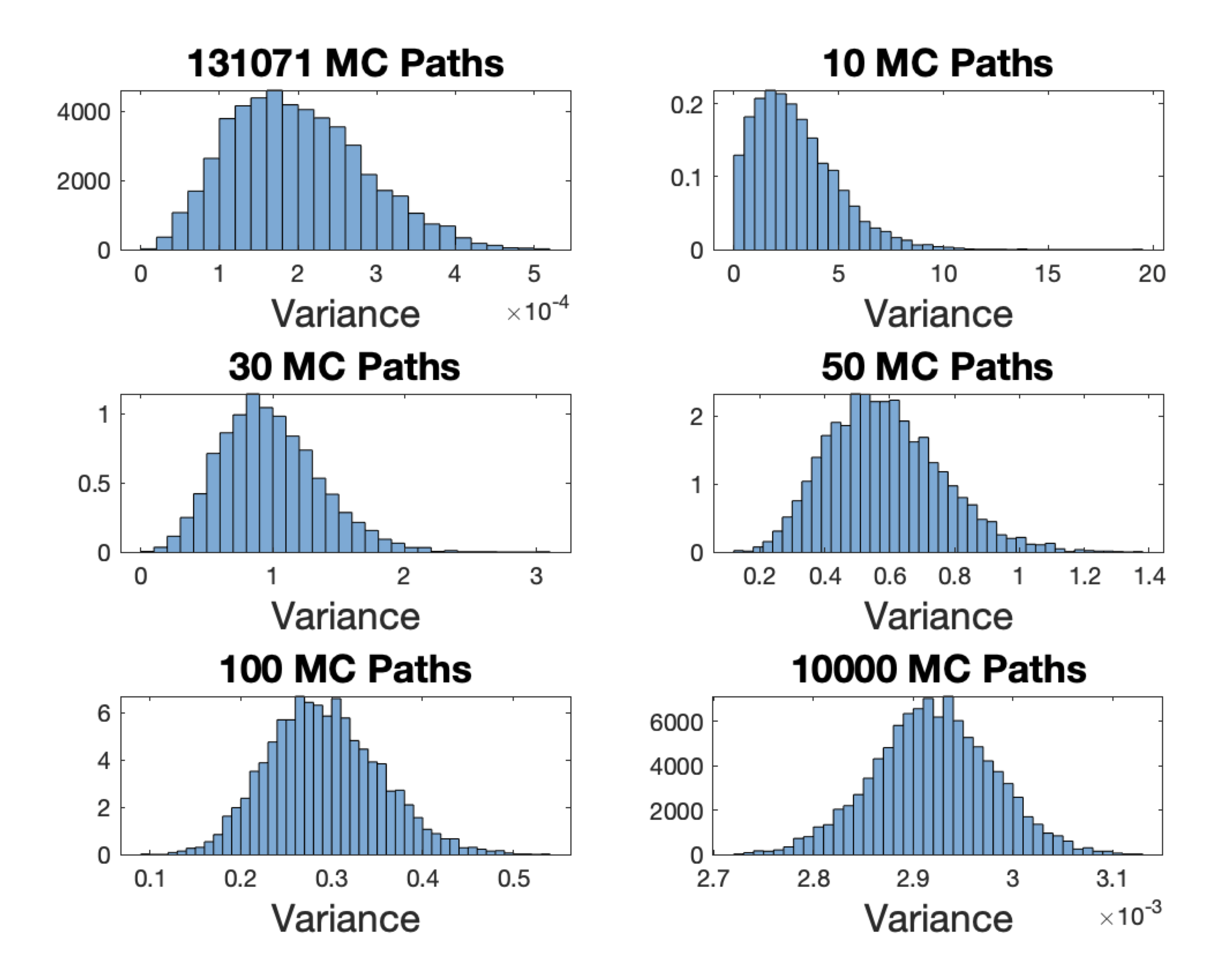}
    \label{fig:mc_var}
\end{figure}

From Figure \ref{fig:fit}, one may infer that we can obtain very accurate price
estimates using 10 paths per outer scenario. Then, assuming that a typical
Nested Monte-Carlo requires 5,000 inner paths to get a reasonable PFE and EE
profile, we can speed up the computation by a factor of 500 using LSMC.
Furthermore, one can argue that a higher polynomial degree could provide a
better fit. Though this may be true when one is dealing with a closed-form
functional approximation, high-degree polynomials tend to overfit data with noise.
These results are most evident in Figure \ref{fig:sse}, where we plot the sum of
squared deterministic residuals (SSE) $\hat{\boldsymbol{\xi}}_{\mathrm{d}}$.

\begin{figure}[hp]
    \centering
    \captionsetup{justification=centering}
    \caption{LSMC Variance Density: \\First time step fit using polynomial of
    degree five}
    \includegraphics[width=0.6\textwidth]{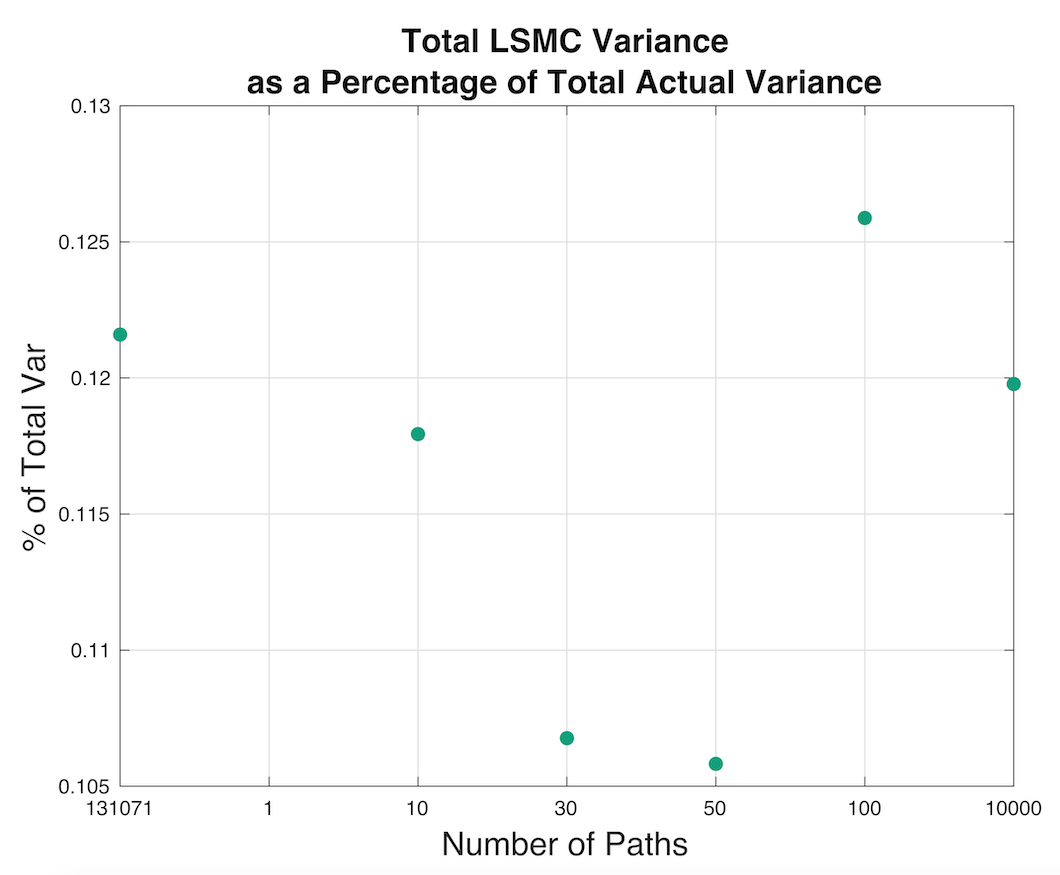}
    \label{fig:lsmc_red}
\end{figure}

Let $\tilde{\mathbf{Y}}$, $\hat{\tilde{\mathbf{Y}}}$,
$\mathbf{Y}_{\mc}$, $\hat{\mathbf{Y}}_{\mc}$, be the Actual
value, the Actual fit, MC price and LSMC price, respectively. Figure
\ref{fig:sse} compares SSE values computed with
\begin{align*}
    \tilde{\boldsymbol{\xi}}_{\mathrm{d}} &= \tilde{\mathbf{Y}} - \hat{\tilde{\mathbf{Y}}},\quad \text{and}   \\
    \hat {\boldsymbol{\xi}}_{\mathrm{d}} &= \mathbf{Y}_{\mc} -
    \tilde{\mathbf{Y}}.
\end{align*}
\begin{landscape}\centering
\vspace*{\fill}
\begin{figure}[hp]
    \captionsetup{justification=centering}
    \caption{Actual vs. LSMC Sum of Squared Errors:\\
     First time step fit with varying degree polynomials}
    \centering\
    \includegraphics[width=\hsize]{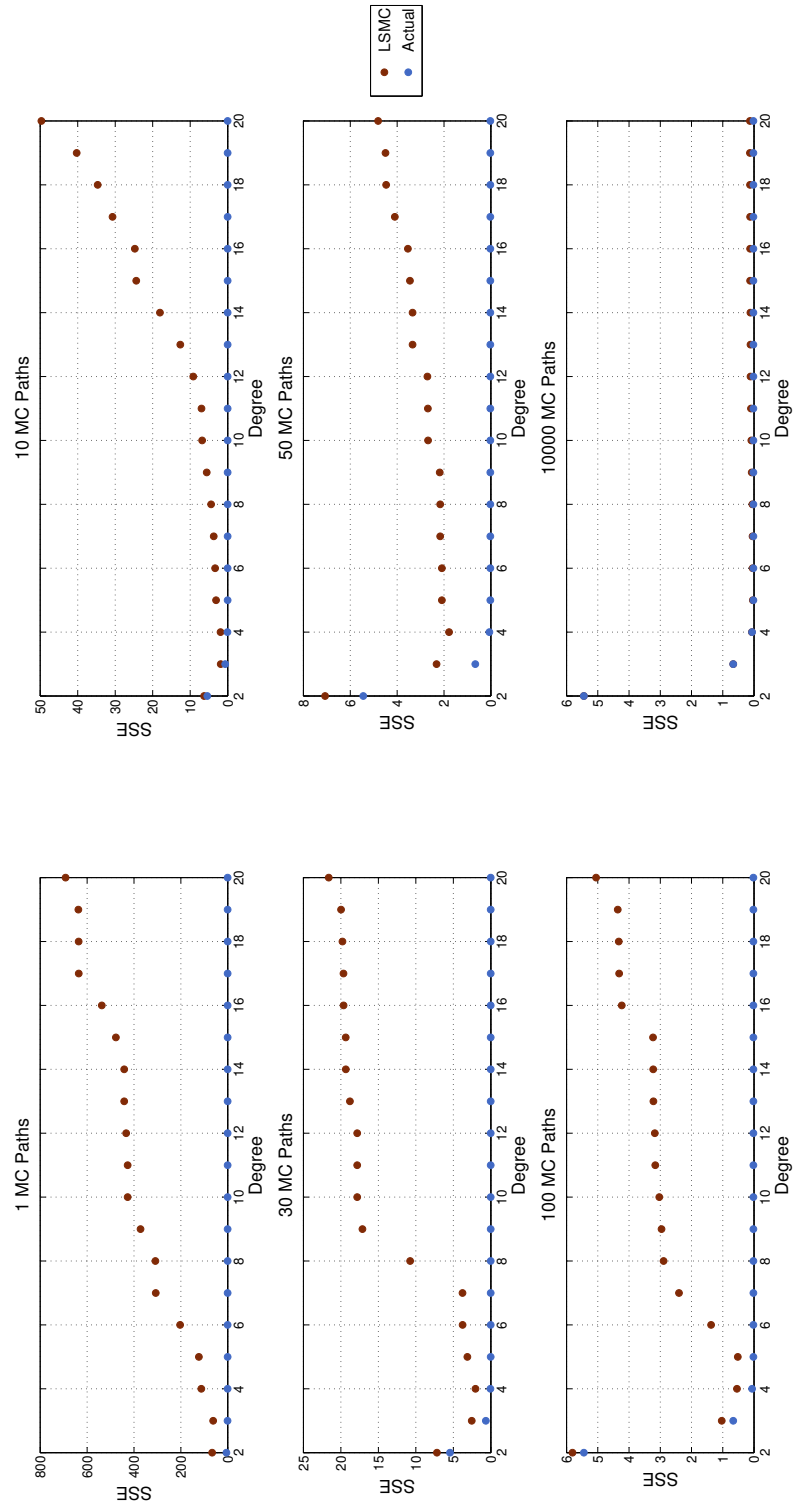}
    \label{fig:sse}
\end{figure}
\vfill
\end{landscape}
We define Actual fit as the least squares fit of price estimates computed using
a large number ($131,072$) of Sobol paths. The SSE of Actual fit that has been
computed using $\tilde{\boldsymbol{\xi}}_\mc$ has a very small (quasi) MC error and
serves as a measure of accuracy. 
From Figure \ref{fig:sse}, we observe monotonic
decrease of SSE when we select higher degree
polynomials for the Actual fit. This does not happen when we fit data with
noise. In general, the less accurate are the price estimates, the sooner SSE
starts to increase. That is, higher degree polynomials tend to overfit data with
noise. With that in mind, we use the third degree polynomial for PFE and EE
estimation: it appears to be the highest degree at which SSE does not increase
in the graphs presented.

\subsection{PFE and EE Results}
\label{sec:pfe_results}
For PFE and EE computations, we use a joint Hull-White two-factor model
for interest rates and Heston stochastic volatility model for FX and
Equity to generate outer, real scenarios. We use different models for pricing,
which vary depending on the instrument. Table \ref{tbl:rn_models} summarizes
the risk-neutral models used in the inner loop for each instrument. For
computational efficiency, with no loss of generality, all instruments mature in
one year and have fixing dates set at 15 day intervals. The profiles obtained
for PFE and EE per instrument are summarized in Figures \ref{fig:PFE} and
\ref{fig:EE}, respectively. 

\begin{table}
\centering
\caption{Risk-neutral models used for pricing of instruments}
\begin{tabular}{lll}
\toprule
\textbf{Instrument} & \textbf{Model}           & \textbf{Underlying} \\
\midrule
Asian   & Hull-White-Black-Scholes  & Foreign Exchange \\
Forward & Geometric Brownian Motion & Equity          \\
Barrier & Hull-White-Black-Scholes  & Foreign Exchange \\
TARN    & Libor Market Model        & Libor Rate       \\
\bottomrule
\end{tabular}
\label{tbl:rn_models}
\end{table}

\begin{figure}
    \captionsetup{justification=centering}
    \caption{Actual vs. LSMC PFE Profile at 95th Percentile}
    \centering
    \includegraphics[width=\textwidth]{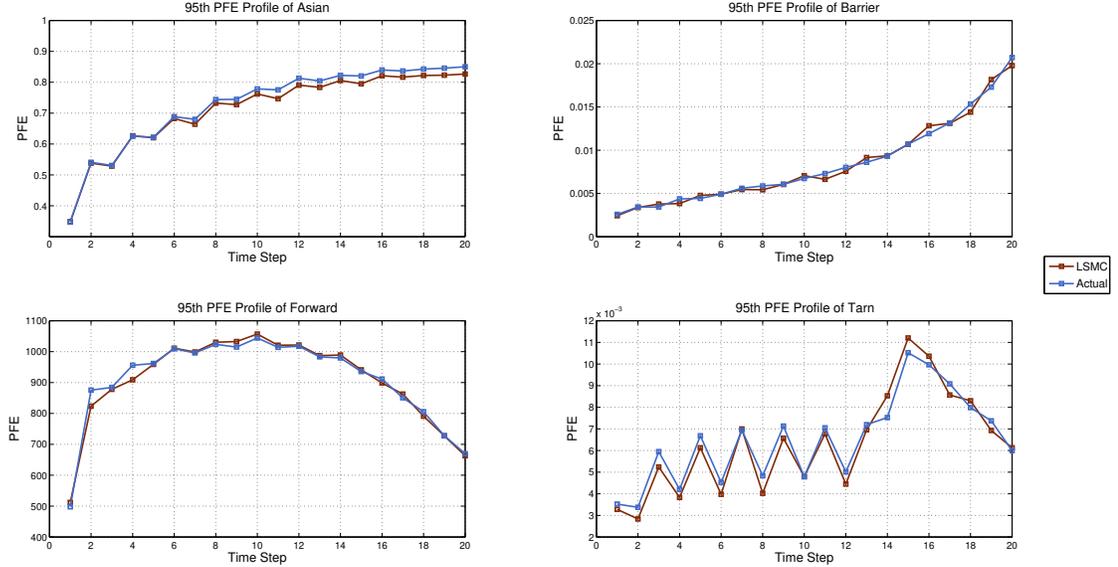}
    \label{fig:PFE}
\end{figure}
\begin{figure}
    \captionsetup{justification=centering}
    \caption{Actual vs. LSMC EE Profile}
    \centering
    \includegraphics[width=\textwidth]{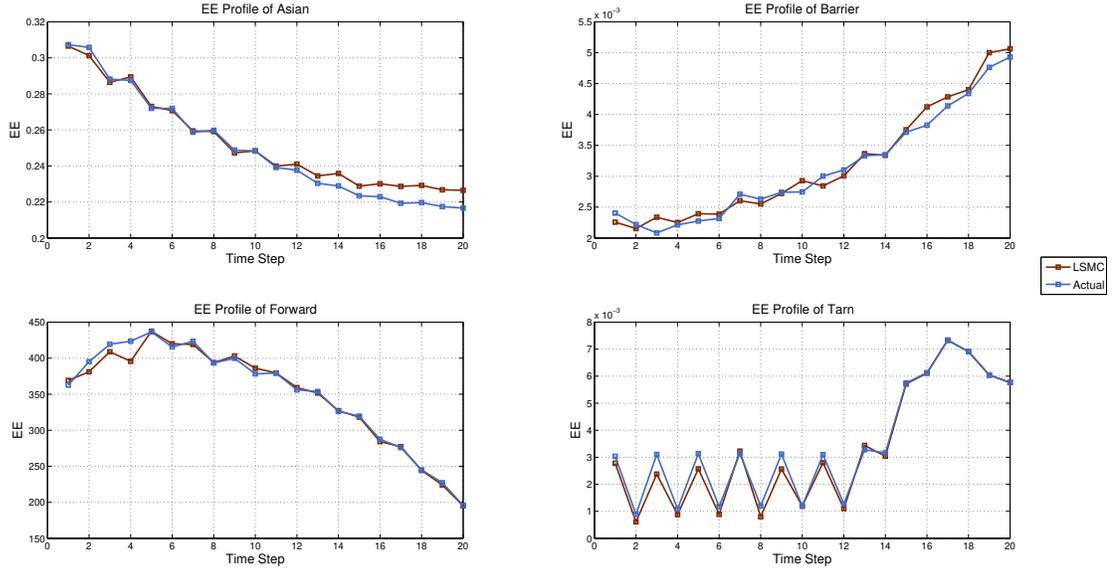}
    \label{fig:EE}
\end{figure}

The PFE and EE profiles for the Actual data series in Figures \ref{fig:PFE} and
\ref{fig:EE} are obtained using 4096 Sobol paths in the inner loop. For the
LSMC PFE and EE profiles, we used varying numbers of paths that were determined
by numerical experiments. We found that smooth instruments such as Asian,
TARN and Accumulative Forward require smaller numbers of paths. In most cases, we
were able to obtain fairly accurate price estimates using anywhere between 10 
and 30 inner paths. To get accurate LSMC Barrier price estimates, we used
anywhere between 30 and 64 inner paths. On average, we were able to speed up the
computation of PFE and EE profiles by a factor of 60. 

\section{Conclusion}\label{sec:conclusion}
In this paper, we were able to show that the least squares Monte-Carlo approach
proposed by Barrie and Hibbert, can capture the tails of the distribution
with a high degree of accuracy. This method, however, does call for numerical
experiments for model calibration where an appropriate number of inner paths and
basis functions need to be selected. Though we were able to capture all
instrument prices with third degree polynomial basis functions, one should exercise
care and diligence when selecting a higher degree. As mentioned previously,
high-degree polynomials may overfit data with noise, which would result in worse
price estimates.

In this paper, we have also proved that the total variance of the least squares
estimates is no greater than the total variance of the original MC estimates.
Indeed, this result should be anticipated since the LSMC method can be viewed as
a smoothing method. Our current results suggest that LSMC should perform better
in aggregate risk metrics such as CVA and/or CVA sensitivity, which is a topic
for a future study. 

\section*{Acknowledgement}
We express special thanks to Ian Isoce for checking the main result and the time
spent in numerous discussions. We also thank Raymond Lee for helping to produce
numerical results, and Alejandra Premat for suggesting credit risk metrics. 

\clearpage
\renewcommand{\theequation}{\thesection.\arabic{equation}}
\appendix

\section{Least Squares Linear Regression} \label{sec:Linear_Regression}
In this section we review the linear regression model and a reduction of
variance theorem. Linear regression modelling and estimation techniques can be
found in many linear regression textbooks such as \cite{seber_lee_lin_reg}.  We
start with a linear regression model of the form 
\begin{equation}
  \label{eq:reg}
  \mathbf{Y} = \mathbf{X} \boldsymbol{\beta} +
  \boldsymbol{\xi}
\end{equation}
with a response variable $\mathbf{Y}$, regression variables $\mathbf{X}$ and
error terms $\boldsymbol{\xi}$. The least squares estimator
$\hat{\boldsymbol{\beta}}$ is the solution to 
\begin{equation}
    \label{eq:min_problem}
    \underset{\boldsymbol{\beta} \in \mathbb{R}^{n}}{\min} || \mathbf{Y} -
    \mathbf{X}\boldsymbol{\beta} ||^{2},
\end{equation}
where $||\mathbf{u}||$ is $L_2$-norm of vector $\mathbf{u}$. The solution to
problem \eqref{eq:min_problem} is given by $\hat{\boldsymbol{\beta}} =
(\mathbf{X^T X})^{-1} \mathbf{X^TY}$, and the estimator of $\mathbf{Y}$ is 
\begin{align}
  \hat{\mathbf{Y}} &= \mathbf{X (X^{T}X)^{-1} X^T Y} \nonumber\\
  &\equiv  \mathbf{H Y}, \label{eq:hat}
\end{align}
provided that $\mathbf{X^TX}$ is invertible.

$\mathbf{H}$, commonly known as the hat matrix, maps the response variables
$\mathbf{Y}$ onto $\hat{\mathbf{Y}} = \mathbf{X} \hat{\boldsymbol{\beta}}$ by
minimizing the sum of squared errors
$\boldsymbol{\xi}^{\mathbf{T}}\boldsymbol{\xi}$.  $\mathbf{H}$ is also the
orthogonal projection onto the column space of $\mathbf{X}$. Clearly,
$\mathbf{H}$ is symmetric and idempotent. 

We can also look at the problem \eqref{eq:reg} from a statistical point of view
and compute the variance of our mean estimator $\hat {\mathbf{Y}}$.
Suppose $\boldsymbol{\xi}$ is a random vector with finite variance, then the
covariance of the mean response $\hat {\mathbf{Y}}$ is given by
\begin{align} 
    \cov(\hat{\mathbf{Y}}) &= \cov(\mathbf{HY}) \nonumber\\
    &= \mathbf{H}\cov(\mathbf{Y})\mathbf{H} \nonumber\\
    &= \mathbf{H} \mathbf{\Sigma} \mathbf{H}    \label{eq:var_mean}, 
\end{align}
where $\mathbf{\Sigma}$ is the covariance matrix of $\boldsymbol{\xi}$.
Furthermore, we can conclude that the total variance of the mean response
estimator $\hat{\mathbf{Y}}$ is less than or equal to the total variance of the
original vector $\mathbf{Y}$ by the following theorem. 
\begin{thm*}
    Let $\mathbf{Y}: \Omega \rightarrow \mathbb{R}^m$ be a random vector having
    a finite variance, and let $\mathbf{H}$ be an orthogonal projection onto a
    linear subspace of $\mathbb{R}^m$. Then,
    \begin{equation*}
        \tr(\cov(\mathbf{HY})) \leq \tr(\cov(\mathbf{Y})),
    \end{equation*}
    where $\tr(\cdot)$ denotes the trace operator. 
\end{thm*}
\begin{proof}
    Expressing $\mathbf{Y} = \mathbf{H Y} + \mathbf{(I - H) Y}$, then 
    \begin{align*}
        \tr\big(\cov(\mathbf{Y})\big) &= \tr\big( \cov( \mathbf{HY + (I - H)Y} ) \big)\\
        & = \tr\big( \mathbf{H} \cov(\mathbf{Y}) \mathbf{H} + \mathbf{(I - H)}\cov(\mathbf{Y})\mathbf{(I - H)} + \mathbf{H} \cov(\mathbf{Y}) \mathbf{(I - H)} + \mathbf{(I - H)}\cov(\mathbf{Y}) \mathbf{H} \big)\\
        & = \tr\big( \mathbf{H} \cov(\mathbf{Y}) \mathbf{H}\big) +
        \tr\big(\mathbf{(I - H)}\cov(\mathbf{Y})\mathbf{(I - H)}\big) \\
        & \qquad \quad + \tr\big(\mathbf{H} \cov(\mathbf{Y})
        \mathbf{(I - H)}\big) + \tr\big(\mathbf{(I - H)}\cov(\mathbf{Y}) \mathbf{H} \big)\\
        & = \tr\big(  \cov(\mathbf{HY}) \big) + \tr\big( \cov(\mathbf{(I - H)Y}) \big) + 2 \tr \big( \mathbf{H} \cov(\mathbf{Y}) \mathbf{(I - H)} \big) \\
        & = \tr\big(  \cov(\mathbf{HY}) \big) + \tr\big( \cov(\mathbf{(I - H)Y}) \big) + 2 \tr \big( \cov(\mathbf{Y}) \mathbf{(I - H)} \mathbf{H} \big) \\
        & = \tr\big(  \cov(\mathbf{HY}) \big) + \tr\big( \cov(\mathbf{(I - H)Y}) \big) + \boldsymbol{0}\\
        & \geq \tr\big(\cov(\mathbf{HY})\big)
    \end{align*}
    by properties of $\mathbf{H}$ and $\tr(\cdot)$ operator.
\end{proof}
\vfill
\section{Definition of Instruments}\label{app:def_inst}

\subsection{Notation}
\begin{itemize}[leftmargin=2cm]
    \item $S_t$ denotes the price of the underlying at time $t$.
    \item $\tau$ denotes a stopping time.
    \item $\tau \wedge t :=  \min (\tau,\, t)$.
    \item $T$ is maturity of the instrument unless otherwise stated.
    \item $T_i < T_j \leq T \quad \forall\, i < j$.
    \item $T_0 = 0$.
\end{itemize}

\subsection{Payoff Functions}
\begin{itemize}[leftmargin=2cm]
\item K-Strike Arithmetic Asian Option with weights $\{w_i\}_{i = 1}^{n}$:
    \begin{align*}
        \text{put} &= (K  - \sum_{i = 1}^{n} { w_i S_{T_i} })^{+} \\
        \text{call} &= (\sum_{i = 1}^{n} { w_i S_{T_i} } - K)^{+}
    \end{align*}
\item  K-Strike Up-and-Out Barrier Option with barrier $B > S_0$ and rebate $R$:
    \begin{align*}
        \text{put} &= (K  - S_{T})^{+}\, \mathbb{I}\{\tau > T\} +
        R\, \mathbb{I}\{\tau \leq T\} \\
        \text{call} &= (S_{T} - K)^{+}\, \mathbb{I}\{\tau > T\} +
        R\, \mathbb{I}\{\tau \leq T\}\\
    \end{align*}
    where $\tau = \min\{t : S_t \geq B\}$. The rebate is paid at time $\tau$ and
    $(S_{T} - K)^{+}$ is paid at maturity.
\item Accumulating Forward Contract with n-payments and barrier $B$:
    \begin{align*}
        \text{short}  &= \sum_{i = 1}^{n}(\alpha |U_i| + \beta
        |D_i|)(K - S_{T_i \wedge \tau}) \\
        \text{long}  &= \sum_{i = 1}^{n}(\alpha |U_i| + \beta
        |D_i|)(S_{T_i \wedge \tau} - K) \\
    \end{align*}
    where
    \[
         \begin{array}{rl}
            U_i &:= \{S_{\tau \wedge t} > K : T_{i - 1} < (\tau \wedge t) < T_i\} \\
            D_i &:= \{S_{\tau \wedge t} \leq K : T_{i - 1} < (\tau \wedge t) < T_i\} \\
            \alpha,\, \beta &\in \mathbb{R}\\
            |A|& \text{denotes the cardinality of set $A$}
        \end{array}
    \]
\item Target Accrual Redemption Note with barrier $B$ on accrual payments, LIBOR
    rate $S$ and fixed rates $K_1$ and $K_2$:
    \begin{align*}
        \text{receiver} &= \sum_{i = 1}^{n}( K_1 - S_{T_i}\,
        \mathbb{I}\{\tau > T_i\} - K_2\, \mathbb{I}\{\tau \leq T_i\}) \\
        \text{payer} &= \sum_{i = 1}^{n}(  S_{T_i}\,
        \mathbb{I}\{\tau > T_i\} + K_2\, \mathbb{I}\{\tau \leq T_i\} - K_1)
    \end{align*}
    where $\tau = \min \{T_i : \sum_{i = 1}^n S_{T_i} \geq B \}$.
\end{itemize}

\newpage
\bibliographystyle{unsrt}

\end{document}